\documentclass[aps,pra,reprint,superscriptaddress]{revtex4-1}
\usepackage{amsfonts}
\usepackage{amsmath}
\usepackage{amsthm}
\usepackage{tikz}
\usetikzlibrary{automata}
\usepackage[caption=false]{subfig}

\usepackage[breaklinks]{hyperref}

\newtheorem{theorem}{Theorem}
\newtheorem{lemma}{Lemma}

\DeclareMathOperator{\poly}{poly}
\DeclareMathOperator{\polylog}{polylog}

\begin{document}

\title{Quantum speedup of the Travelling Salesman Problem for bounded-degree graphs}

\author{Dominic J. Moylett}
\email[Corresponding author: ]{dominic.moylett@bristol.ac.uk}
\affiliation{Quantum Engineering Technology Labs, H. H. Wills Physics Laboratory and Department of Electrical \& Electronic Engineering, University of Bristol, BS8 1FD, UK}
\affiliation{Quantum Engineering Centre for Doctoral Training, H. H. Wills Physics Laboratory and Department of Electrical \& Electronic Engineering, University of Bristol, BS8 1FD, UK}
\affiliation{Heilbronn Institute for Mathematical Research, University of Bristol, BS8 1SN, UK}
\author{Noah Linden}
\email{n.linden@bristol.ac.uk}
\affiliation{School of Mathematics, University of Bristol, BS8 1TW, UK}
\author{Ashley Montanaro}
\email{ashley.montanaro@bristol.ac.uk}
\affiliation{School of Mathematics, University of Bristol, BS8 1TW, UK}

\date{\today}

\begin{abstract}
The Travelling Salesman Problem is one of the most famous problems in graph theory. However, little is currently known about the extent to which quantum computers could speed up algorithms for the problem. In this paper, we prove a quadratic quantum speedup when the degree of each vertex is at most $3$ by applying a quantum backtracking algorithm to a classical algorithm by Xiao and Nagamochi. We then use similar techniques to accelerate a classical algorithm for when the degree of each vertex is at most $4$, before speeding up higher-degree graphs via reductions to these instances.
\end{abstract}

\maketitle


\section{Introduction}

A salesman has a map of $n$ cities that they want to visit, including the roads between the cities and how long each road is. Their aim is to start at their home, visit each city and then return home. To avoid wasting time, they want to visit each city exactly once and travel via the shortest route. So what route should the salesman take?

This is an instance of the Travelling Salesman Problem (TSP). More generally, this problem takes an undirected graph $G = (V, E)$ of $n$ vertices connected by $m$ weighted edges and returns the shortest cycle which passes through every vertex exactly once, known as a Hamiltonian cycle, if such a cycle exists. If no Hamiltonian cycle exists, we should report that no Hamiltonian cycle has been found. The length or cost of an edge is given by an $n \times n$ matrix $C = (c_{ij})$ of positive integers, known as a cost matrix. This problem has a number of applications, ranging from route finding as in the story above to circuit board drilling~\cite{grotschel1991}.

Unfortunately, the salesman might have to take a long time in order to find the shortest route. The TSP has been shown to be NP-hard~\cite[Chapter $3$]{lawler1985}, suggesting that even the best algorithms for exactly solving it must take a superpolynomial amount of time. Nevertheless, the importance of the problem has motivated a substantial amount of classical work to develop algorithms for solving it provably more efficiently than the na\"ive algorithm which checks all $O((n-1)!)$ of the potential Hamiltonian cycles in the graph. Here we consider whether these algorithms can be accelerated using quantum computational techniques.

Grover's famous quantum algorithm~\cite{grover96} for fast unstructured search can be applied to the na\"ive classical algorithm to achieve a runtime of $O(\sqrt{n!})$, up to polynomial terms in $n$. However, the best classical algorithms are already substantially faster than this. For many years, the algorithm with the best proven worst-case bounds for the general TSP was the Held-Karp algorithm~\cite{held1962}, which runs in $O(n^22^n)$ time and uses $O(n2^n)$ space. This algorithm uses the fact that for any shortest path, any subpath visiting a subset of vertices on that path must be the shortest path for visiting those vertices. Held and Karp used this to solve the TSP by computing the length of the optimal route for starting at vertex $1$, visiting every vertex in a set $S \subseteq V$ and finishing at a vertex $l \in S$. Denoting the length of this optimal route $D(S, l)$, they showed that this distance could be computed as
\[
D(S, l) = \begin{cases} c_{1l} & \text{if } S = \{l\}\\
          \min_{m \in S \setminus \{l\}}\left[D(S \setminus \{l\}, m) + c_{ml}\right] & \text{otherwise.}
  \end{cases}
\]
Solving this relation recursively for $S=V$ would result in iterating over all $O((n-1)!)$ Hamiltonian cycles again, but Held and Karp showed that the relation could be solved in $O(n^22^n)$ time using dynamic programming. Bj{\"o}rklund et al.\ \cite{bjorklund2008} developed on this result, showing that modifications to the Held-Karp algorithm could yield a runtime of
\[ O((2^{k + 1} - 2k - 2)^{n/(k + 1)}\poly(n)), \]
where $k$ is the largest degree of any vertex in the graph; this bound is strictly less than $O(2^n)$ for all fixed $k$. Unfortunately, it is not known whether quantum algorithms can accelerate general dynamic programming algorithms. Similarly, it is unclear whether TSP algorithms based around the standard classical techniques of branch-and-bound~\cite{little1963} or branch-and-cut~\cite{padberg1991} are amenable to quantum speedup.



Here we apply known quantum-algorithmic techniques to accelerate more recent classical TSP algorithms for the important special case of bounded-degree graphs. We say that a graph $G$ is degree-$k$ if the maximal degree of any vertex in $G$ is at most $k$. A recent line of research has produced a sequence of algorithms which improve on the $O^*(2^n)$ runtime of the general Held-Karp algorithm in this setting, where the notation $O^*(c^n)$ omits polynomial factors in $n$. First, Eppstein presented algorithms which solve the TSP on degree-3 graphs in time $O^*(2^{n/3}) \approx O^*(1.260^n)$, and on degree-4 graphs in time $O^*((27/4)^{n/3}) \approx O^*(1.890^n)$~\cite{eppstein2007}. The algorithms are based on the standard classical technique of {\em backtracking}, an approach where a tree of partial solutions is explored to find a complete solution to a problem (see Section \ref{sec:backtrack} for an introduction to this technique). Following subsequent improvements~\cite{iwama07,liskiewicz14}, the best classical runtimes known for algorithms based on this general approach are $O^*(1.232^n)$ for degree-3 graphs~\cite{xiao2016degree3}, and $O^*(1.692^n)$ for degree-4 graphs~\cite{xiao2016degree4}, in each case due to Xiao and Nagamochi. All of these algorithms use polynomial space in $n$.

An algorithm of Bodlaender et al.~\cite{bodlaender15} achieves a faster runtime of $O^*(1.219^n)$ for solving the TSP in degree-3 graphs, which is the best known; however, this algorithm uses exponential space. Similarly, an algorithm of Cygan et al.~\cite{cygan11} solves the TSP in unweighted degree-4 graphs in $O^*(1.588^n)$ time and exponential space. Both of these algorithms use an approach known as cut-and-count, which is based on dynamic programming, so a quantum speedup is not known for either algorithm.

In the case where we have an upper bound $L$ on the maximum edge cost in the graph, Bj\"orklund~\cite{bjorklund14} gave a randomised algorithm which solves the TSP on arbitrary graphs in $O^*(1.657^n L)$ time and polynomial space, which is an improvement on the runtime of the Xiao-Nagamochi algorithm for degree-4 graphs when $L$ is subexponential in $n$. Again, the techniques used in this algorithm do not seem obviously amenable to quantum speedup.

Here we use a recently developed quantum backtracking algorithm~\cite{montanaro2015} to speed up the algorithms of Xiao and Nagamochi in order to find Hamiltonian cycles shorter than a given upper bound, if such cycles do exist. We run this algorithm several times, using binary search to specify what our upper bound should be, in order to find the shortest Hamiltonian cycle and solve the Travelling Salesman Problem. In doing so, we achieve a near-quadratic reduction in the runtimes:

\begin{theorem}
There are bounded-error quantum algorithms which solve the TSP on degree-3 graphs in time $O^*(1.110^n \log L \log \log L)$ and on degree-4 graphs in time $O^*(1.301^n \log L \log \log L)$, where $L$ is the maximum edge cost. The algorithms use $\poly(n)$ space.
\label{thm:deg34}
\end{theorem}

In this result and elsewhere in the paper, ``bounded-error'' means that the probability that the algorithm either doesn't find a Hamiltonian cycle when one exists or returns a non-optimal Hamiltonian cycle is at most $1/3$. This failure probability can be reduced to $\delta$, for arbitrary $\delta > 0$, by repeating the algorithm $O(\log 1/\delta)$ times. Also here and throughout the paper, $\log$ denotes $\log$ base 2. Note that the time complexity of our algorithms has some dependence on $L$, the largest edge cost in the input graph. However, this dependence is quite mild. For any graph whose edge costs are specified by $w$ bits, $L \le 2^w$. Thus terms of the form $\polylog(L)$ are at most polynomial in the input size.

Next, we show that degree-5 and degree-6 graphs can be dealt with via a randomised reduction to the degree-4 case.

\begin{theorem}
\label{thm:deg6}
There is a bounded-error quantum algorithm which solves the TSP on degree-5 and degree-6 graphs in time $O^*(1.680^n\log L \log \log L)$. The algorithm uses $\poly(n)$ space.
\end{theorem}

We summarise our results in Table \ref{tab:summary}.

\begin{table*}
\begin{center}
\begin{tabular}{|c|c|c|c|}
\hline Degree & Quantum & Classical (poly space) & Classical (exp space) \\
\hline 3 & $O^*(1.110^n \polylog L)$ & $O^*(1.232^n)$~\cite{xiao2016degree3} & $O^*(1.219^n)$~\cite{bodlaender15} \\
 4 & $O^*(1.301^n \polylog L)$ & $O^*(1.692^n)$~\cite{xiao2016degree4}, $O^*(1.657^n L)$~\cite{bjorklund14} & $O^*(1.588^n)$~\cite{cygan11}\\
 5, 6 & $O^*(1.680^n \polylog L)$ & $O^*(1.657^n L)$~\cite{bjorklund14} & --- \\
\hline
\end{tabular}
\end{center}
\caption{Runtimes of our quantum algorithms for a graph of $n$ vertices with maximum edge cost $L$, compared with the best classical algorithms known.}
\label{tab:summary}
\end{table*}

\subsection{Related work}

Surprisingly little work has been done on quantum algorithms for the TSP. D\"orn \cite{dorn2007} proposed a quantum speedup for the TSP for degree-3 graphs by applying amplitude amplification \cite{brassard1997} and quantum minimum finding~\cite{durr1996} to Eppstein's algorithm, and stated a quadratic reduction in the runtime. However, we were not able to reproduce this result (see Section~\ref{sec:backtrack} below for a discussion).

Very recently, Mandr{\`a}, Guerreschi and Aspuru-Guzik~\cite{mandra2016} developed a quantum algorithm for finding a Hamiltonian cycle in time $O(2^{(k-2)n/4})$ in a graph where {\em every} vertex has degree $k$. Their approach reduces the problem to an Occupation problem, which they solve via a backtracking process accelerated by the quantum backtracking algorithm~\cite{montanaro2015}. The bounds obtained from their algorithm are $O(1.189^n)$ for $k = 3$ and $O(1.414^n)$ for $k=4$, in each case a bit slower than the runtimes of our algorithms; for $k \ge 5$, their algorithm has a slower runtime than Bj\"orklund's classical algorithm~\cite{bjorklund14}.

Marto\v{n}\'ak, Santoro and Tosatti \cite{martonak2004} explored the option of using quantum annealing to find approximate solutions for the TSP. Rather than solve the problem purely through quantum annealing, they simplify their Ising Hamiltonian for solving the TSP and use path-integral Monte Carlo \cite{barker1979} to run their model. While no bounds on run time or accuracy were strictly proven, they concluded by comparing their algorithm to simulated annealing via the Metropolis-Hastings algorithm \cite{metropolis1953} and the Kernighan-Lin algorithm for approximately solving the TSP \cite{kernighan1970}. Their results showed that ad hoc algorithms could perform better than general simulated or quantum annealing, but quantum annealing could outperform simulated annealing alone. However, they noted that simulated annealing could perform better than in their analysis if combined with local search heuristics~\cite{martin1996}.

Chen et al.\ \cite{chen11} experimentally demonstrated a quantum annealing algorithm for the TSP. Their demonstration used a nuclear-magnetic-resonance quantum simulator to solve the problem for a graph with 4 vertices.

\subsection{Organisation}

We start by introducing the main technique we use, backtracking, and comparing it with amplitude amplification. Then, in Section \ref{sec:bd}, we describe how this technique can be used to accelerate classical algorithms of Xiao and Nagamochi for graphs of degree at most 4~\cite{xiao2016degree3,xiao2016degree4}. In Section \ref{sec:higher-bound}, we extend this approach to graphs of degree at most 6.



\section{Backtracking algorithms for the TSP}
\label{sec:backtrack}

Many of the most efficient classical algorithms known for the TSP are based around a technique known as backtracking. Backtracking is a general process for solving constraint satisfaction problems, where we have $v$ variables and we need to find an assignment to these variables such that they satisfy a number of constraints. A na\"{i}ve search across all possible assignments will be inefficient, but if we have some local heuristics then we can achieve better performance by skipping assignments that will definitely fail.

Suppose each variable can be assigned one value from $[d] := \{0, \dots,d-1\}$. We define the set of partial assignments for $v$ variables as $\mathcal{D} := (\{1,\dots,v\}, [d])^j$, where $j \leq v$, with the first term denoting the variable to assign and the second denoting the value it is assigned. Using this definition for partial assignments, backtracking algorithms have two components. The first is a predicate, $P:\mathcal{D} \rightarrow \{\text{true}, \text{false}, \text{indeterminate}\}$, which takes a partial assignment and returns true if this assignment will definitely result in the constraints being satisfied regardless of how everything else is assigned, false if the assignment will definitely result in the constraints being unsatisfied, and indeterminate if we do not yet know. The second is a heuristic, $h:\mathcal{D} \rightarrow \{1,\dots,v\}$, which takes a partial assignment and returns the next variable to assign.

The following simple recursive classical algorithm takes advantage of $P$ and $h$ to solve a constraint satisfaction problem. We take as input a partial assignment (initially, the empty assignment). We run $P$ on this partial assignment; if the result is true then we return the partial assignment, and if it is false then we report that no solutions were found in this recursive call. We then call $h$ on this partial assignment and find out what the next variable to assign is. For every value in $i \in [d]$ we can assign that variable, we recursively call the backtracking algorithm with $i$ assigned to that variable. If one of the recursive calls returns a partial assignment then we return that assignment, otherwise we report that no solutions were found in this call. We can view this algorithm as exploring a tree whose vertices are labelled with partial assignments. The size of the tree determines the worst-case runtime of the algorithm, assuming that there is no assignment that satisfies all the constraints.

It is known that this backtracking algorithm can be accelerated using quantum techniques:

\begin{theorem}[Montanaro \cite{montanaro2015}]
\label{thm:backtrack}
Let $\mathcal{A}$ be a backtracking algorithm with predicate $P$ and heuristic $h$ that finds a solution to a constraint satisfaction problem on $v$ variables by exploring a tree of at most $T$ vertices. There is a quantum algorithm which finds a solution to the same problem with failure probability $\delta$ with $O(\sqrt{T}v^{3/2}\log v\log(1/\delta))$ uses of $P$ and $h$.
\end{theorem}

Montanaro's result is based on a previous algorithm by Belovs \cite{belovs2013,belovs13a}, and works by performing a quantum walk on the backtracking tree to find vertices corresponding to assignments which satisfy the constraints. The reader familiar with \cite{montanaro2015} may note that the definition of the set of partial assignments $\mathcal{D}$ is different to that given there, in that it incorporates information about the ordering of assignments to variables. However, it is easy to see from inspection of the algorithm of \cite{montanaro2015} that this change does not affect the stated complexity of the algorithm.

It is worth noting that more standard quantum approaches such as amplitude amplification~\cite{brassard1997} will not necessarily achieve a quadratic speedup over the classical backtracking algorithm. Amplitude amplification requires access to a function $f:\{0,1\}^k \rightarrow \{\text{true}, \text{false}\}$ and a guessing function $\mathcal{G}$. If the probability of $\mathcal{G}$ finding a result $x \in \{0,1\}^k$ such that $f(x) = \text{true}$ is $p$, then amplitude amplification will succeed after $O(1/\sqrt{p})$ applications of $f$ and $\mathcal{G}$~\cite{brassard1997}.

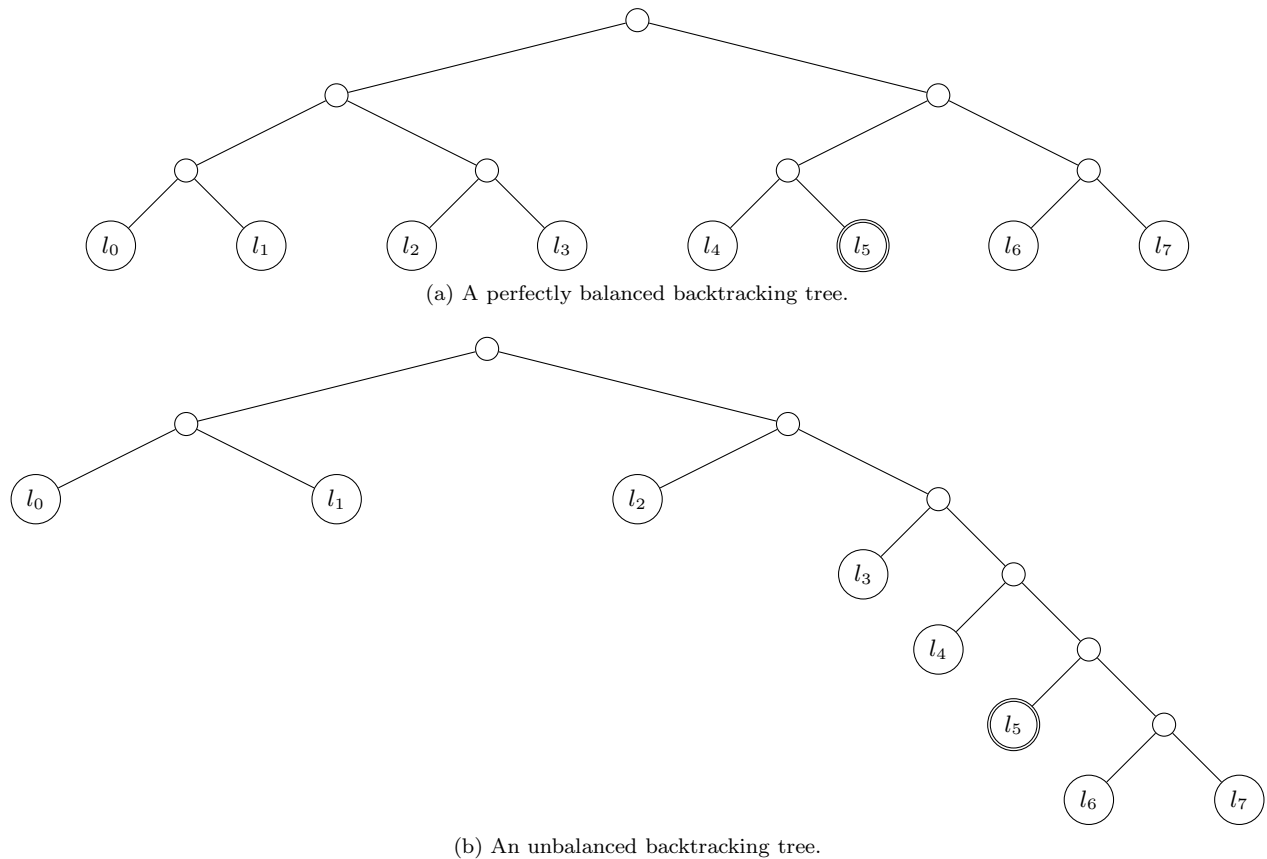
\begin{figure*}
\subfloat[A perfectly balanced backtracking tree.\label{fig:balanced-tree}]{
\begin{tikzpicture}
\tikzstyle{vertex}=[draw,shape=circle]
\path (0,0) node[vertex](b0){} (-4,-1) node[vertex](b1){} (4,-1) node[vertex](b2){} (-6,-2) node[vertex](b3){} (-2,-2) node[vertex](b4){} (2,-2) node[vertex](b5){} (6,-2) node[vertex](b6){} (-7,-3) node[vertex](b7){$l_0$} (-5,-3) node[vertex](b8){$l_1$} (-3,-3) node[vertex](b9){$l_2$} (-1,-3) node[vertex](b10){$l_3$} (1,-3) node[vertex](b11){$l_4$} (3,-3) node[vertex, accepting](b12){$l_5$} (5,-3) node[vertex](b13){$l_6$} (7,-3) node[vertex](b14){$l_7$};
\draw (b7) -- (b3) -- (b1) -- (b0) -- (b2) -- (b6) -- (b14);
\draw (b8) -- (b3);
\draw (b9) -- (b4) -- (b1);
\draw (b10) -- (b4);
\draw (b11) -- (b5) -- (b2);
\draw (b12) -- (b5);
\draw (b13) -- (b6);
\end{tikzpicture}
}
\hfill
\subfloat[An unbalanced backtracking tree.\label{fig:unbalanced-tree}]{
\begin{tikzpicture}
\tikzstyle{vertex}=[draw,shape=circle]
\path (0,0) node[vertex](b0){} (-4,-1) node[vertex](b1){} (4,-1) node[vertex](b2){} (-6,-2) node[vertex](b3){$l_0$} (-2,-2) node[vertex](b4){$l_1$} (2,-2) node[vertex](b5){$l_2$} (6,-2) node[vertex](b6){} (5,-3) node[vertex](b7){$l_3$} (7,-3) node[vertex](b8){} (6,-4) node[vertex](b9){$l_4$} (8,-4) node[vertex](b10){} (7,-5) node[vertex, accepting](b11){$l_5$} (9,-5) node[vertex](b12){} (8,-6) node[vertex](b13){$l_6$} (10,-6) node[vertex](b14){$l_7$};
\draw (b3) -- (b1) -- (b0) -- (b2) -- (b6) -- (b8) -- (b10) -- (b12) -- (b14);
\draw (b4) -- (b1);
\draw (b5) -- (b2);
\draw (b7) -- (b6);
\draw (b9) -- (b8);
\draw (b11) -- (b10);
\draw (b13) -- (b12);
\end{tikzpicture}
}
\caption{Example backtracking trees, where $l_5$ is a leaf corresponding to a solution to a constraint satisfaction problem. In the perfectly balanced case of Fig.\ \ref{fig:balanced-tree}, each leaf can be associated with a 3-bit string corresponding to a path to that leaf. But in the unbalanced case of Fig.\ \ref{fig:unbalanced-tree}, specifying a path to a leaf requires 6 bits.}
\label{fig:tree}
\end{figure*}

To apply amplitude amplification, we would need to access the leaves of the tree, as these are the points where the backtracking algorithm is certain whether or not a solution will be found. Thus, for each integer $i$, we would need to find a way of determining the $i$'th leaf $l_i$ in the backtracking tree. In the case of a perfectly balanced tree, such as Fig.\ \ref{fig:balanced-tree}, where every vertex in the tree is either a leaf or has exactly $d$ branches descending from it, such a problem is easy: write $i$ in base $d$ and use each digit of $i$ to decide which branch to explore. But not all backtracking trees are perfectly balanced, such as in Fig.\ \ref{fig:unbalanced-tree}. In these cases, finding leaf $l_i$ is hard as we cannot be certain which branch leads to that leaf. Some heuristic approaches, by performing amplitude amplification on part of the tree, can produce better speedups for certain trees, but do not provide a general speedup on the same level as the quantum backtracking algorithm~\cite{montanaro2015}.

It is also worth understanding the limitations of the quantum backtracking algorithm, and why it cannot necessarily speed up all algorithms termed ``backtracking algorithms''~\cite{montanaro2015}. First, a requirement for the quantum algorithm is that decisions made in one part of the backtracking tree are independent of results in another part of the tree, which is not true of all classical algorithms, such as constraint recording algorithms \cite{dechter1990}. Second, the runtime of the quantum algorithm depends on the size of the entire tree. Thus, to achieve a quadratic speedup over a classical algorithm, the algorithm must explore the whole backtracking tree, instead of stopping after finding the first solution or intelligently skipping branches such as in backjumping \cite{dechter1990}. Therefore, it is important to check on a case-by-case basis whether classical backtracking algorithms can actually be accelerated using Theorem \ref{thm:backtrack}.

Another limitation of the quantum backtracking algorithm is that often there will be a metric $M:\mathcal{D} \rightarrow \mathbb{N}$ we want the backtracking algorithm to minimise while satisfying the other constraints. This is particularly relevant for the TSP, where the aim is to return the shortest Hamiltonian cycle. Classical backtracking algorithms can achieve this by recursively travelling down each branch of the tree to find results $D_1,\dots,D_d \in \mathcal{D}$ and returning the result that minimises $M$. The quantum backtracking algorithm cannot perform this; it instead returns a solution selected randomly from the tree that satisfies the constraints. In order to achieve a quantum speedup when finding the result that minimises $M$, we can modify the original predicate to prune results which are greater than or equal to a given bound. We then repeat the algorithm in a binary search fashion, updating our bound based on whether or not a solution was found. This will find the minimum after repeating the quantum algorithm at most $O(\log M_{max})$ times, where 
\[M_{max} = \max\{M(D):D\in \mathcal{D}, P(D) = \text{true}\}.\]
We describe this binary search approach in more detail in Sec.\ \ref{sec:deg3speedup}.

The intuition behind why backtracking is a useful technique for solving the TSP is that we can attempt to build up a Hamiltonian cycle by determining for each edge in the graph whether it should be included in the cycle (``forced''), or deleted from the graph. As we add more edges to the cycle, we may either find a contradiction (e.g.\ produce a non-Hamiltonian cycle) or reduce the graph to a special case that can be handled efficiently (e.g.\ a collection of disjoint cycles of four unforced edges). This can sometimes allow us to prune the backtracking tree substantially.

To analyse the performance of backtracking algorithms for the TSP, a problem size measure is often defined that is at least 0 and at most $n$ (e.g.\ the number of vertices minus the number of forced edges). Note that if there are more than $n$ forced edges then it is impossible to form a Hamiltonian cycle that includes every forced edge, so the number of forced edges is at most $n$. At the start of the backtracking algorithm, there are no forced edges so the problem size is $n$. Each step of the backtracking algorithm reduces the problem size until the size is $0$, at which point either the $n$ forced edges form a Hamiltonian cycle or a Hamiltonian cycle that includes every forced edge cannot be found. A quasiconvex program can be developed based on how the backtracking algorithm reduces the problem size. Solving this quasiconvex problem produces a runtime in terms of the problem size, which can be re-written in terms of $n$ due to the problem size being at most $n$.

It was proposed by D\"orn \cite{dorn2007} that amplitude amplification could be applied to speed up the runtime of Eppstein's algorithm for the TSP on degree-3 graphs~\cite{eppstein2007} from $O^*(2^{n/3})$ to $O^*(2^{n/6})$. Amplitude amplification can be used in this setting by associating a bit-string with each sequence of choices of whether to force or delete an edge, and searching over bit-strings to find the shortest valid Hamiltonian cycle. However, as suggested by the general discussion above, a difficulty with this approach is that some branches of the recursion, as shown in Fig.~\ref{fig:size-decrease-by-two}, only reduce the problem size by 2 (as measured by the number of vertices $n$, minus the number of forced edges). The longest branch of the recursion can, as a result, be more than $n/3$ levels deep. In the worst case, this depth could be as large as $n/2$ levels. Specifying the input to the checking function $f$ could then require up to $n/2$ bits, giving a search space of size $O(2^{n/2})$. Under these conditions, searching for the solution via amplitude amplification could require up to $O^*(2^{n/4})$ time in the worst case. To yield a better runtime, we must take more of an advantage of the structure of our search space to avoid instances which will definitely not succeed.

The same issue with amplitude amplification applies to other classical algorithms for the TSP which are based on backtracking~\cite{xiao2016degree3,xiao2016degree4}. In the case of the Xiao-Nagamochi algorithm for degree-3 graphs, although the overall runtime bound proven for the problem means that the number of vertices in the tree is $O(2^{3n/10})$, several of the branching vectors used in their analysis have branches that reduce the problem size by less than $10/3$, leading to a branch in the tree that could be more than $3n/10$ levels deep.

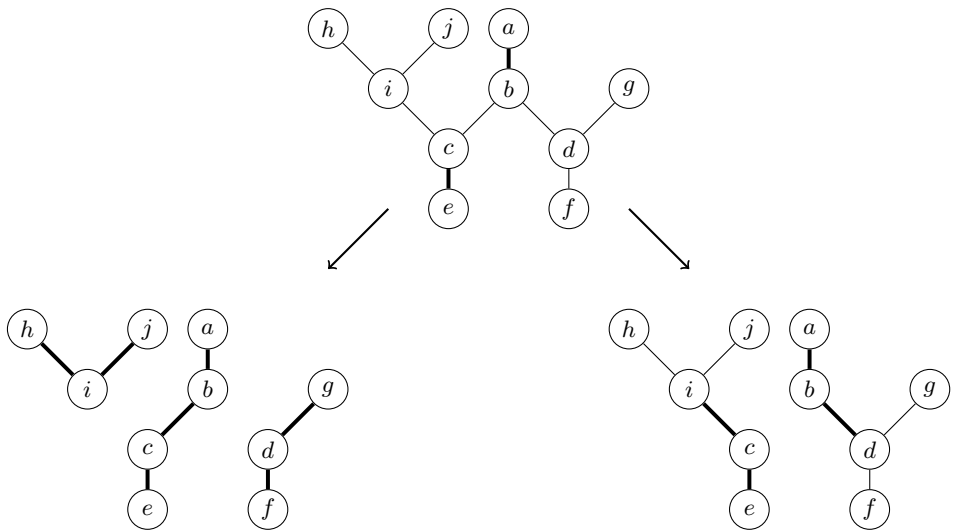
\begin{figure*}
\begin{center}
\begin{tikzpicture}[scale=0.8]
\tikzstyle{vertex}=[draw,shape=circle,inner sep=0pt,minimum size=15pt]
\path (0,0) node[vertex](f0){$a$};
\path (-2,-1) node[vertex](x0){$i$} (-1,-2) node[vertex](x1){$c$} (0,-1) node[vertex](f1){$b$} (1,-2) node[vertex](y1){$d$} (2,-1) node[vertex](y0){$g$};
\path (-1,-3) node[vertex](x2){$e$} (1,-3) node[vertex](y2){$f$};
\path (-3,0) node[vertex](x3){$h$} (-1,0) node[vertex](x4){$j$};
\draw[line width=1.5pt] (f0) -- (f1);
\draw[line width=1.5pt] (x2) -- (x1);
\draw (x0) -- (x1) -- (f1) -- (y1) -- (y2);
\draw (x3) -- (x0) -- (x4);
\draw (y0) -- (y1);

\path (-5,-5) node[vertex](f0){$a$};
\path (-7,-6) node[vertex](x0){$i$} (-6,-7) node[vertex](x1){$c$} (-5,-6) node[vertex](f1){$b$} (-4,-7) node[vertex](y1){$d$} (-3,-6) node[vertex](y0){$g$};
\path (-6,-8) node[vertex](x2){$e$} (-4,-8) node[vertex](y2){$f$};
\path (-8,-5) node[vertex](x3){$h$} (-6,-5) node[vertex](x4){$j$};
\draw[line width=1.5pt] (f0) -- (f1) -- (x1) -- (x2);
\draw[line width=1.5pt] (x3) -- (x0) -- (x4);
\draw[line width=1.5pt] (y2) -- (y1) -- (y0);

\path (5,-5) node[vertex](f0){$a$};
\path (3,-6) node[vertex](x0){$i$} (4,-7) node[vertex](x1){$c$} (5,-6) node[vertex](f1){$b$} (6,-7) node[vertex](y1){$d$} (7,-6) node[vertex](y0){$g$};
\path (4,-8) node[vertex](x2){$e$} (6,-8) node[vertex](y2){$f$};
\path (2,-5) node[vertex](x3){$h$} (4,-5) node[vertex](x4){$j$};
\draw[line width=1.5pt] (f0) -- (f1) -- (y1);
\draw[line width=1.5pt] (x0) -- (x1) -- (x2);
\draw (x3) -- (x0) -- (x4);
\draw (y0) -- (y1) -- (y2);

\draw[thick, ->] (-2,-3) -- (-3,-4);
\draw[thick, ->] (2,-3) -- (3,-4);
\end{tikzpicture}
\end{center}
\caption{An instance of the recursive step in Eppstein's backtracking algorithm for the TSP~\cite{eppstein2007} for a subgraph of a larger graph $G$, with forced edges displayed in bold and  branching on edge $bc$. If we force $bc$, then $b$ and $c$ are both incident to two forced edges, so $bd$ and $ci$ cannot be part of the Hamiltonian cycle and can be removed from the graph. After these edges are removed, vertices $i$ and $d$ are both of degree $2$, so in order to reach those vertices the edges $hi$, $ij$, $df$ and $dg$ must also be included in the Hamiltonian cycle. So forcing $bc$ has overall added five edges to the Hamiltonian cycle. On the other hand, if we remove edge $bc$, we find that $b$ and $c$ are vertices of degree $2$, so edges $bd$ and $ci$ must be part of the Hamiltonian cycle. Thus we have only added two more edges to the Hamiltonian cycle.
\label{fig:size-decrease-by-two}}
\end{figure*}


\section{Quantum speedups for the Travelling Salesman Problem on bounded-degree graphs \label{sec:bd}}
\label{sec:deg3}

Our algorithms are based on applying the quantum algorithm for backtracking (Theorem \ref{thm:backtrack}) to Xiao and Nagamochi's algorithm~\cite{xiao2016degree3}. Before describing our algorithms, we need to introduce some terminology from~\cite{xiao2016degree3} and describe their original algorithm. The algorithm, and its analysis, are somewhat involved, so we omit details wherever possible.

\subsection{The algorithm of Xiao and Nagamochi}
\label{sec:xndeg3}

A graph $G$ is $k$-edge connected if there are $k$ edge-disjoint paths between every pair of vertices. An edge in $G$ is said to be forced if it must be included in the final tour, and unforced otherwise. The set of forced edges is denoted $F$, and the set of unforced edges is denoted $U$. An induced subgraph of unforced edges which is maximal and connected is called a $U$-component. If a $U$-component is just a single vertex, then that $U$-component is trivial. A maximal sequence $\mathcal{C}$ of edges in a $U$-component $H$ is called a circuit if either:
\begin{itemize}
\item $\mathcal{C} = \{xy\}$ and there are three edge-disjoint paths from $x$ to $y$,
\item or $\mathcal{C} = \{c_0, c_1,\dots,c_{m-1}\}$ such that for $0 \leq i < m-1$, there is a subgraph $B_i$ of $H$ such that the only two unforced edges incident to $B_i$ are $c_i$ and $c_{i+1}$.
\end{itemize}

A circuit is reducible if subgraph $B_i$ for some $i$ is incident to only two edges. In order for $B_i$ to be reached, both edges incident to $B_i$ need to be forced. Forcing one edge in the circuit then means that the other edges can be either forced or removed. The polynomial time and space process by Xiao and Nagamochi to reduce circuits, by forcing and removing alternating edges in the circuit, is known as the {\em circuit procedure} \cite{xiao2016degree3}.

Note that each edge can be in at most one circuit. If two distinct circuits $\mathcal{C}, \mathcal{C}'$ shared an edge $e_i$, then there are two possibilities. The first is that there is a subgraph $B_i$ incident to unforced edges $e_i \in \mathcal{C} \cap \mathcal{C}', e_{i+1} \in \mathcal{C} - \mathcal{C}', e_j \in \mathcal{C}' - \mathcal{C}$. In this case, $B_i$ is incident to more than two unforced edges, so neither $\mathcal{C}$ nor $\mathcal{C}'$ are circuits, which is a contradiction.

The second is that there is some edge $e_i$ which is incident to distinct subgraphs $B_i, B_i'$ related to $\mathcal{C}, \mathcal{C}'$, respectively. Circuits are maximal sequences, so it cannot be the case that $B_i$ is a subgraph of $B_i'$, otherwise $\mathcal{C}' \subseteq \mathcal{C}$. Now we consider the subgraphs $B_i \cap B_i'$ and $B_i - B_i'$, which must be connected by unforced edges as they are both subgraphs of $B_i$. These unforced edges are incident to $B_i'$, which is a contradiction as they are not part of $\mathcal{C}'$.

Let $X$ be a subgraph. We define $\text{cut}(X)$ to be the set of edges that connect $X$ to the rest of the graph. If $|\text{cut}(X)| = 3$, then we say that $X$ is $3$-cut reducible. It was shown by Xiao and Nagamochi~\cite{xiao2016degree3} that, if $X$ is 3-cut reducible, $X$ can be replaced with a single vertex of degree $3$ with outgoing edges weighted such that the length of the shortest Hamiltonian cycle is preserved.

The definition of $4$-cut reducible is more complex. Let $X$ be a subgraph such that $\text{cut}(X) \subseteq F$ and $|\text{cut}(X)| = 4$. A solution to the TSP would have to partition $X$ into two disjoint paths such that every vertex in $X$ is in one of the two paths. If $x_1, x_2, x_3$ and $x_4$ are the four vertices in $X$ incident to the four edges in $\text{cut}(X)$, then there are three ways these paths could start and end:
\begin{itemize}
\item $x_1 \leftrightarrow x_2$ and $x_3 \leftrightarrow x_4$,
\item $x_1 \leftrightarrow x_3$ and $x_2 \leftrightarrow x_4$,
\item or $x_1 \leftrightarrow x_4$ and $x_2 \leftrightarrow x_3$.
\end{itemize}
We say that $X$ is $4$-cut reducible if for at least one of the above cases it is impossible to create two disjoint paths in $X$ that include all vertices in $X$. Xiao and Nagamochi defined a polynomial time and space process for applying the above reductions, known as {\em $3/4$-cut reduction}~\cite{xiao2016degree3}.

A set of edges $\{e_i\}$ are {\em parallel} if they are incident to the same vertices (note that here we implicitly let $G$ be a multigraph; these may be produced in intermediate steps of the algorithm). If there are only two vertices in the graph, then the TSP can be solved directly by forcing the shortest two edges. Otherwise if at least one of the edges is not forced, then we can reduce the problem by removing the longer unforced edges until the vertices are only adjacent via one edge. This is the process Xiao and Nagamochi refer to as {\em eliminating parallel edges} \cite{xiao2016degree3}.

Finally, a graph is said to satisfy the parity condition if every $U$-component is incident to an even number of forced edges and for every circuit $\mathcal{C}$, an even number of the corresponding subgraphs $B_i$ satisfy that $|\text{cut}(B_i) \cap F|$ is odd.




We are now ready to describe Xiao and Nagamochi's algorithm. The algorithm takes as input a graph $G = (V, E)$ and a set of forced edges $F \subseteq E$ and returns the length of the shortest Hamiltonian cycle in $G$ containing all the edges in $F$, if one exists.

The algorithm is based on four subroutines: {\em eliminating parallel edges}, the {\em 3/4-cut reduction}, {\em selecting a good circuit} and the {\em circuit procedure}, as well as the following lemma:

\begin{lemma}[Eppstein~\cite{eppstein2007}]
\label{lem:trivial}
If every $U$-component in a graph $G$ is trivial or a component of a 4-cycle, then a minimum cost tour can be found in polynomial time.
\end{lemma}

We will not define the subroutines here in any detail; for our purposes, it is sufficient to assume that they all run in polynomial time and space. The circuit procedure for a circuit $\mathcal{C}$ begins by either adding an edge $e \in \mathcal{C}$ to $F$ or deleting it from the graph, then performing some other operations. ``Branching on a circuit $\mathcal{C}$ at edge $e \in \mathcal{C}$'' means generating two new instances from the current instance by applying each of these two variants of the circuit procedure starting with $e$.

The Xiao-Nagamochi algorithm, named $\text{TSP3}$, proceeds as follows, reproduced from~\cite{xiao2016degree3}:

\begin{enumerate}
\item {\bf If} $G$ is not $2$-edge-connected or the instance violates the parity condition, then return $\infty$;
\item {\bf Elseif} there is a reducible circuit $\mathcal{C}$, then return $\text{TSP3}(G', F')$ for an instance $(G',F')$ obtained by applying the circuit procedure on $\mathcal{C}$ started by adding a reducible edge in $\mathcal{C}$ to $F$;
\item {\bf Elseif} there is a pair of parallel edges, then return $\text{TSP3}(G',F')$ for an instance $(G',F')$ obtained by applying the reduction rule of eliminating parallel edges;
\item {\bf Elseif} there is a $3/4$-cut reducible subgraph $X$ containing at most eight vertices, then return $\text{TSP3}(G',F')$ for an instance $(G',F')$ obtained by applying the $3/4$-cut reduction on $X$;
\item {\bf Elseif} there is a $U$-component $H$ that is neither trivial nor a $4$-cycle, then select a good circuit $\mathcal{C}$ in $H$ and return $\min\{\text{TSP3}(G_1,F_1), \text{TSP3}(G_2,F_2)\}$, where $(G_1,F_1)$ and $(G_2,F_2)$ are the
two resulting instances after branching on $\mathcal{C}$;
\item {\bf Else} [each $U$-component of the graph is trivial or a $4$-cycle], solve the problem directly in polynomial time by Lemma \ref{lem:trivial} and return the cost of an optimal tour.
\end{enumerate}

Step $1$ of the algorithm checks that the existence of a Hamiltonian cycle is not ruled out, by ensuring that that there are at least two disjoint paths between any pair of vertices and that the graph satisfies the parity condition. Step 2 reduces any reducible circuit by initially forcing one edge and then alternately removing and forcing edges. Step $3$ removes any parallel edges from the graph, and step $4$ removes any circuits of three edges as well as setting up circuits of four edges so that all edges incident to them are forced. Step 5 is the recursive step, branching on a good circuit by either forcing or removing an edge in the circuit and then applying the circuit procedure. The algorithm continues these recursive calls until it either finds a Hamiltonian cycle or $G \setminus F$ is a collection of single vertices and cycles of length $4$, all of which are disjoint from one another, at which point the problem can be solved in polynomial time via step $6$.

Xiao and Nagamochi looked at how the steps of the algorithm, and step $5$ in particular as the branching step, reduced the size of the problem for different graph structures. From this they derived a quasiconvex program corresponding to $19$ branching vectors, each describing how the problem size is reduced at the branching step in different circumstances. Analysis of this quasiconvex program showed that the algorithm runs in $O^*(2^{3n/10})$ time and polynomial space \cite{xiao2016degree3}.

\subsection{Quantum speedup of the Xiao-Nagamochi algorithm}
\label{sec:deg3speedup}

Here we describe how we apply the quantum backtracking algorithm to the Xiao-Nagamochi algorithm. It is worth noting that the quantum backtracking algorithm will not necessarily return the shortest Hamiltonian cycle, but instead returns a randomly selected Hamiltonian cycle that it found. Adding constraints on the length of the Hamiltonian cycles to our predicate and running the quantum backtracking algorithm multiple times will allow us to find a solution to the TSP.

The first step towards applying the quantum backtracking algorithm is to define the set of partial assignments. A partial assignment will be a list of edges in $G$ ordered by when they are assigned in the backtracking algorithm and paired with whether the assignment was to force or remove the edge. The assignment is denoted $A \in (\{1,\dots,m\}, \{\text{force}, \text{remove}\})^j$, where $j \leq m$. We have $m \le 3n/2$ as $G$ is degree-3.


The quantum approach to backtracking requires us to define a predicate $P$ and heuristic $h$, each taking as input a partial assignment. Our predicate and heuristic make use of a reduction function, introduced in \cite{xiao2016degree3}, as a subroutine; this function is described in the next subsection. However it may be worth noting that the algorithm uses the original graph $G$, and partial assignments of it at each stage.

Firstly, we describe the $P$ function, which takes a partial assignment $A = ((e_1, A_1),\dots,(e_j, A_j))$ as input:

\begin{enumerate}
\item Using the partial assignment $A$, apply the reduction function to $(G, F)$ to get $(G', F')$.
\item If $G'$ is not $2$-edge-connected or fails the parity condition, then return false.
\item If every $U$-component in $G'$ is either trivial or a $4$-cycle, then return true.
\item Return indeterminate.
\end{enumerate}

Step $2$ matches step $1$ of Xiao and Nagamochi's algorithm. Step $3$ is where the same conditions are met as in step $6$ of Xiao and Nagamochi's algorithm, where a shortest length Hamiltonian cycle is guaranteed to exist and can be found in polynomial time classically via Lemma \ref{lem:trivial}. Step $4$ continues the branching process, which together with how the circuit is picked by $h$ and step $2$(c) of the reduction function (qv) matches step $5$ of Xiao and Nagamochi.

The $h$ function is described as follows, taking as input a partial assignment $A = ((e_1, A_1),\dots,(e_j, A_j))$ of the edges of $G$:

\begin{enumerate}
\item Using the partial assignment $A$, apply the reduction function to $(G, F)$ to get $(G', F')$.
\item Select a $U$-component in $G'$ that is neither trivial nor a cycle of length $4$. Select a circuit $\mathcal{C}$ in that component that fits the criteria of a ``good'' circuit~\cite{xiao2016degree3}, then select an edge $e_i' \in \mathcal{C}$.
\item Return an edge in $G$ corresponding to $e_i'$ (if there is more than one, choosing one arbitrarily).
\end{enumerate}

Step $2$ applies step $5$ of Xiao and Nagamochi's algorithm, by selecting the next circuit to branch on and picking an edge in that circuit. If the reduced version of the graph results in $h$ picking an edge corresponding to multiple edges in the original graph, step $3$ ensures that we only return one of these edges to the backtracking algorithm, as step $2$(b) of the reduction function will ensure that every edge in the original graph corresponding to an edge in the reduced graph will be consistently forced or removed. The rest of the circuit will be forced or removed by step $2$(c) of the reduction function.

We can now apply the backtracking algorithm (Theorem \ref{thm:backtrack}) to $P$ and $h$ to find a Hamiltonian cycle. We will later choose its failure probability $\delta$ to be sufficiently small that we can assume that it always succeeds, i.e.\ finds a Hamiltonian cycle if one exists, and otherwise reports that one does not exist. At the end of the algorithm, we will receive either the information that no assignment was found, or a partial assignment. By applying the reduction steps and the partial assignments, we can reconstruct the graph at the moment our quantum algorithm terminated, which will give a graph such that every $U$-component is either trivial or a 4-cycle. We then construct and return the full Hamiltonian cycle in polynomial time using step $6$ of Xiao and Nagamochi's algorithm~\cite{xiao2016degree3}. 

To solve the TSP, we need to find the shortest Hamiltonian cycle. This can be done as follows. First, we run the backtracking algorithm. If the backtracking algorithm does not return a Hamiltonian cycle then we report that no Hamiltonian cycle was found. Otherwise after receiving Hamiltonian cycle $\Gamma$ with length $L_\Gamma$, we create variables $\ell \leftarrow 0$ \& $u \leftarrow L_\Gamma$ and modify $P$ to return false if
\[\sum_{e_{i,j}\in F}c_{ij} \geq \lceil(\ell + u)/2\rceil.\]
If no cycle is found after running the algorithm again, we set $\ell \leftarrow \lceil(\ell + u)/2\rceil$ and repeat. Otherwise, upon receiving Hamiltonian cycle $\Gamma'$ with total cost $L_{\Gamma'}$, we set $u \leftarrow L_{\Gamma'}$ and repeat. We continue repeating until $\ell$ and $u$ converge, at which point we return the Hamiltonian cycle found by the algorithm. In the worst case scenario, where the shortest cycle is found during the first run of the backtracking algorithm, this algorithm matches a binary search. So the number of repetitions of the backtracking algorithm required to return the shortest Hamiltonian cycle is at most $O(\log L')$, where
%
%
\begin{align}
L' = \sum_{i = 1}^{n}\max \{c_{ij} : j \in \{1,\dots,n\} \}
\label{eqn:l}
\end{align}
is an upper bound on the total cost of any Hamiltonian cycle in the graph.

\subsection{The reduction function}
\label{sec:reduction}

Finally, we describe the reduction function, which takes the original graph $G$ and partial assignment $A$, and applies the partial assignment to this graph in order to reduce it to a smaller graph $G'$ with forced edges $F'$. This reduction might mean that forcing or removing a single edge in $G'$ would be akin to forcing several edges in $G$. For example, let $X$ be a $3$-reducible subgraph of at most $8$ vertices with $\text{cut}(X) = \{ax_1, bx_2, cx_3\}$ for vertices $x_1, x_2, x_3 \in V(X)$. The $3/4$-cut reduction reduces $X$ to a single vertex $x \in G'$ with edges $ax, bx, cx$. If the edges $ax$ and $bx$ are forced, this is equivalent to forcing every edge in $\Pi \cup \{ax_1, bx_2\}$, where $\Pi$ is the shortest path that starts at $x_1$, visits every vertex in $X$ exactly once, and ends at $x_2$. As we need to solve the problem in terms of the overall graph $G$ and not the reduced graph $G'$, our assigned variables need to correspond to edges in $G$. To do this, our heuristic includes a step where if the edge selected in $G'$ corresponds to multiple edges in $G$, we simply select one of the corresponding edges in $G$ to return. Likewise, if the next edge in our partial assignment is one of several edges in $G$ corresponding to a single edge in $G'$, we apply the same assignment to all of the other corresponding edges in $G$.

The reduction function works as follows, using reductions and procedures from Xiao and Nagamochi \cite{xiao2016degree3}:

\begin{enumerate}
\item Create a copy of the graph $G' \leftarrow G$ and set of forced edges $F' \leftarrow \emptyset$.
\item For each $i=1,\dots,j$:
\begin{enumerate}
\item Repeat until none of the cases apply:
\begin{enumerate}
\item If $G'$ contains a reducible circuit $\mathcal{C}$, then apply the circuit procedure to $\mathcal{C}$.
\item If $G'$ contains parallel edges, then apply the reduction rule of eliminating parallel edges.
\item If $G'$ contains a subgraph $X$ of at most $8$ vertices such that $X$ is $3/4$-cut reducible, then apply the $3$/$4$-cut reduction to $X$.
\end{enumerate}
\item Apply assignment $(e_i, a_i)$ to $(G', F')$ by adding edge $e_i$ to $F'$ if $a_i = \text{force}$, or deleting edge $e_i$ from $G'$ if $A_i = \text{remove}$. If edge $e_i$ is part of a set of edges corresponding to a single edge in $G'$, apply the same assignment to all edges in $G$ which correspond to the same edge in $G'$ by adding them all to $F'$ if $a_i = \text{force}$, or deleting them all from $G'$ if $a_i = \text{remove}$.
\item Apply the circuit procedure to the rest of the circuit containing edge $e_i$.
\end{enumerate}
\item Run step 2(a) again.
\item Return $(G', F')$.
\end{enumerate}

Step $2$(a)i recreates step $2$ from Xiao and Nagamochi's original algorithm by applying the circuit procedure where possible. Step $2$(a)ii recreates step $3$ of the original algorithm by applying the reduction of parallel edges. Step $2$(a)iii recreates step $4$ of the original algorithm via the $3/4$-cut reduction. Step $2$(b) applies the next step of the branching that has been performed so far, to ensure that the order in which the edges are forced is the same as in the classical algorithm. Step $2$(c) corresponds to branching on a circuit at edge $e_i$. Finally, step $3$ checks whether or not the graph can be reduced further by running the reduction steps again.

One might ask if an edge could be part of two circuits, in which case our algorithm would fail as it would not be able to reduce the circuit. However, as discussed in Sec.\ \ref{sec:xndeg3}, any edge can only be part of at most one circuit.

\subsection{Analysis}

Steps $2$(a)i-iii of the reduction algorithm can be completed in polynomial time~\cite{xiao2016degree3}. All of these steps also reduce the size of a problem by at least a constant amount, so only a polynomial number of these steps are needed. Step 2(b) is constant time and step 2(c) can be run in polynomial time as the circuit is now reducible. All steps are only repeated $O(m)$ times, so the whole reduction algorithm runs in polynomial time in terms of $m$.

Steps $2$ and $3$ of the $h$ subroutine run in polynomial time as searching for a good circuit in a component can be done in polynomial time \cite{xiao2016degree3}. Likewise, steps 2 and 3 of the $P$ function involve looking for certain structures in the graph that can be found in polynomial time. As a result, the runtimes for the $P$ and $h$ functions are both polynomial in $m$.

By Theorem \ref{thm:backtrack}, the number of calls to $P$ and $h$ we make in order to find a Hamiltonian cycle with failure probability $\delta$ is $O(\sqrt{T}\poly(m)\log (1/\delta))$, where $T$ is the size of the backtracking tree, which in our case is equal to the number of times the Xiao-Nagamochi algorithm branches on a circuit. $P$ and $h$ both run in polynomial time and as a result can be included in the $\poly(m)$ term of the runtime. Because $m \leq 3n/2$, the polynomial term in this bound is also polynomial in terms of $n$.

The behaviour of the $P$ and $h$ subroutines is designed to reproduce the behaviour of Xiao and Nagamochi's TSP3 algorithm~\cite{xiao2016degree3}. It is shown in~\cite[Theorem 1]{xiao2016degree3} that this algorithm is correct, runs in time $O^*(2^{3n/10})$ and uses polynomial space. As the runtime of the TSP3 algorithm is an upper bound on the number of branching steps it makes, the algorithm branches on a circuit $O^*(2^{3n/10})$ times. Therefore, the quantum backtracking algorithm finds a Hamiltonian cycle, if one exists, with failure probability at most $\delta$ in time $O^*(2^{3n/20} \log(1/\delta)) \approx O^*(1.110^n \log(1/\delta))$ and polynomial space.

Finding the shortest Hamiltonian cycle requires repeating the algorithm $O(\log L')$ times, where $L'$ is given in Equation \ref{eqn:l}. By using a union bound over all the runs of the algorithm, to ensure that all runs succeed with high probability it is sufficient for the failure probability $\delta$ of each run to be at most $O(1/(\log L'))$. From this we obtain the following result, proving the first part of Theorem \ref{thm:deg34}:

\begin{theorem}
There is a bounded-error quantum algorithm which solves the TSP on degree-3 graphs in time $O^*(1.110^n \log L \log \log L)$, where $L$ is the maximum edge cost. The algorithm uses $\poly(n)$ space.
\end{theorem}

Note that we have used the bound $L' \le n L$, where the extra factor of $n$ is simply absorbed into the hidden $\poly(n)$ term.

\section{Extending to higher-degree graphs \label{sec:higher-bound}}

We next consider degree-$k$ graphs for $k \ge 4$. We start with degree-4 graphs by applying the quantum backtracking algorithm to another algorithm by Xiao and Nagamochi~\cite{xiao2016degree4}. We then extend this approach to graphs of higher degree by reducing the problem to degree-4 graphs.

\subsection{Degree-4 graphs}

Here we will show the following, which is the second part of Theorem \ref{thm:deg34}:

\begin{theorem}
There is a bounded-error quantum algorithm which solves the TSP for degree-4 graphs in time $O^*(1.301^n\log L \log \log L)$, where $L$ is the maximum edge cost. The algorithm uses $\poly(n)$ space.
\end{theorem}

As the argument is very similar to the degree-3 case, we only sketch the proof.

\begin{proof}[Proof sketch]
Xiao and Nagamochi's algorithm for degree-4 graphs works in a similar way to their algorithm for degree-3 graphs: The graph is reduced in polynomial time by looking for specific structures in the graph and then picking an edge in the graph to branch on. We apply the quantum backtracking algorithm as before, finding a Hamiltonian cycle with failure probability $\delta$ in $O^*(1.301^n\log(1/\delta))$ time. We then use binary search to find the shortest Hamiltonian cycle after $O(\log L)$ repetitions of the algorithm, rejecting if the total length of the forced edges is above a given threshold. To achieve overall failure probability $1/3$, the algorithm runs in $O^*(1.301^n\log L\log \log L)$ time.
\end{proof}

\subsection{Degree-5 and degree-6 graphs}

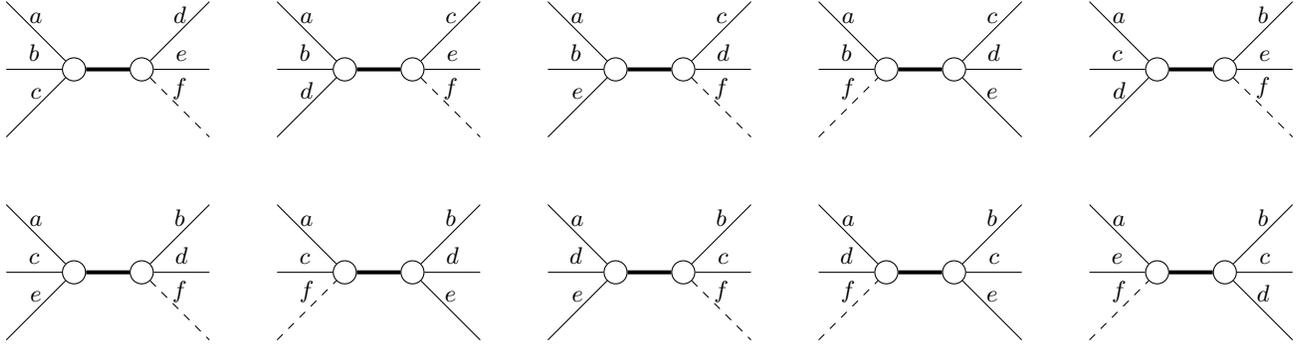
\begin{figure*}
\begin{center}
\begin{tikzpicture}[scale=0.9]
\tikzstyle{vertex}=[draw,shape=circle]
\path (0,0) node[vertex](x1){} (1,0) node[vertex](y1){};
\draw[line width=1.5pt] (x1) -- (y1);
\draw[] (-1,1) -- node[above] {$a$} (x1);
\draw[] (-1,0) -- node[above] {$b$} (x1);
\draw[] (-1,-1) -- node[above] {$c$} (x1);
\draw[] (y1) -- node[above] {$d$} (2,1);
\draw[] (y1) -- node[above] {$e$} (2,0);
\draw[dashed] (y1) -- node[above] {$f$} (2,-1);

\path (4,0) node[vertex](x1){} (5,0) node[vertex](y1){};
\draw[line width=1.5pt] (x1) -- (y1);
\draw[] (3,1) -- node[above] {$a$} (x1);
\draw[] (3,0) -- node[above] {$b$} (x1);
\draw[] (3,-1) -- node[above] {$d$} (x1);
\draw[] (y1) -- node[above] {$c$} (6,1);
\draw[] (y1) -- node[above] {$e$} (6,0);
\draw[dashed] (y1) -- node[above] {$f$} (6,-1);

\path (8,0) node[vertex](x1){} (9,0) node[vertex](y1){};
\draw[line width=1.5pt] (x1) -- (y1);
\draw[] (7,1) -- node[above] {$a$} (x1);
\draw[] (7,0) -- node[above] {$b$} (x1);
\draw[] (7,-1) -- node[above] {$e$} (x1);
\draw[] (y1) -- node[above] {$c$} (10,1);
\draw[] (y1) -- node[above] {$d$} (10,0);
\draw[dashed] (y1) -- node[above] {$f$} (10,-1);

\path (12,0) node[vertex](x1){} (13,0) node[vertex](y1){};
\draw[line width=1.5pt] (x1) -- (y1);
\draw[] (11,1) -- node[above] {$a$} (x1);
\draw[] (11,0) -- node[above] {$b$} (x1);
\draw[dashed] (11,-1) -- node[above] {$f$} (x1);
\draw[] (y1) -- node[above] {$c$} (14,1);
\draw[] (y1) -- node[above] {$d$} (14,0);
\draw[] (y1) -- node[above] {$e$} (14,-1);

\path (16,0) node[vertex](x1){} (17,0) node[vertex](y1){};
\draw[line width=1.5pt] (x1) -- (y1);
\draw[] (15,1) -- node[above] {$a$} (x1);
\draw[] (15,0) -- node[above] {$c$} (x1);
\draw[] (15,-1) -- node[above] {$d$} (x1);
\draw[] (y1) -- node[above] {$b$} (18,1);
\draw[] (y1) -- node[above] {$e$} (18,0);
\draw[dashed] (y1) -- node[above] {$f$} (18,-1);

\path (0,-3) node[vertex](x1){} (1,-3) node[vertex](y1){};
\draw[line width=1.5pt] (x1) -- (y1);
\draw[] (-1,-2) -- node[above] {$a$} (x1);
\draw[] (-1,-3) -- node[above] {$c$} (x1);
\draw[] (-1,-4) -- node[above] {$e$} (x1);
\draw[] (y1) -- node[above] {$b$} (2,-2);
\draw[] (y1) -- node[above] {$d$} (2,-3);
\draw[dashed] (y1) -- node[above] {$f$} (2,-4);

\path (4,-3) node[vertex](x1){} (5,-3) node[vertex](y1){};
\draw[line width=1.5pt] (x1) -- (y1);
\draw[] (3,-2) -- node[above] {$a$} (x1);
\draw[] (3,-3) -- node[above] {$c$} (x1);
\draw[dashed] (3,-4) -- node[above] {$f$} (x1);
\draw[] (y1) -- node[above] {$b$} (6,-2);
\draw[] (y1) -- node[above] {$d$} (6,-3);
\draw[] (y1) -- node[above] {$e$} (6,-4);

\path (8,-3) node[vertex](x1){} (9,-3) node[vertex](y1){};
\draw[line width=1.5pt] (x1) -- (y1);
\draw[] (7,-2) -- node[above] {$a$} (x1);
\draw[] (7,-3) -- node[above] {$d$} (x1);
\draw[] (7,-4) -- node[above] {$e$} (x1);
\draw[] (y1) -- node[above] {$b$} (10,-2);
\draw[] (y1) -- node[above] {$c$} (10,-3);
\draw[dashed] (y1) -- node[above] {$f$} (10,-4);

\path (12,-3) node[vertex](x1){} (13,-3) node[vertex](y1){};
\draw[line width=1.5pt] (x1) -- (y1);
\draw[] (11,-2) -- node[above] {$a$} (x1);
\draw[] (11,-3) -- node[above] {$d$} (x1);
\draw[dashed] (11,-4) -- node[above] {$f$} (x1);
\draw[] (y1) -- node[above] {$b$} (14,-2);
\draw[] (y1) -- node[above] {$c$} (14,-3);
\draw[] (y1) -- node[above] {$e$} (14,-4);

\path (16,-3) node[vertex](x1){} (17,-3) node[vertex](y1){};
\draw[line width=1.5pt] (x1) -- (y1);
\draw[] (15,-2) -- node[above] {$a$} (x1);
\draw[] (15,-3) -- node[above] {$e$} (x1);
\draw[dashed] (15,-4) -- node[above] {$f$} (x1);
\draw[] (y1) -- node[above] {$b$} (18,-2);
\draw[] (y1) -- node[above] {$c$} (18,-3);
\draw[] (y1) -- node[above] {$d$} (18,-4);
\end{tikzpicture}
\end{center}
\caption{Breaking a vertex of degree 5 or 6 into two lower-degree vertices. In the degree-5 case, dashed edge $f$ is not present and the vertex is split into one vertex of degree $3$ and another of degree $4$ connected by a forced edge in bold. In the degree-6 case, dashed edge $f$ is present and the vertex is split into two vertices of degree $4$ connected by a forced edge. If edges $a$ and $b$ are included in the original graph's shortest Hamiltonian cycle, then they must not be adjacent to one another in the final graph. This holds in six of the ten ways of splitting the vertex. \label{fig:degree-5}}
\end{figure*}

To deal with degree-5 and degree-6 graphs, we reduce them to the degree-4 case. The complexity of the two cases turns out to be the same; however, for clarity we consider each case separately.

\begin{theorem}
\label{thm:deg5}
There is a bounded-error quantum algorithm which solves the TSP for degree-5 graphs in time $O^*(1.680^n\log L\log \log L)$.
\end{theorem}

\begin{proof}
Our algorithm works by splitting each vertex of degree 5 into one vertex of degree $3$ and another of degree $4$ connected by a forced edge. The forced edges can be included in our quantum algorithm by modifying step 1 of the reduction function so that $F'$ contains all the forced edges created by splitting a vertex of degree-$5$ into two vertices connected by a forced edge. Once all degree-$5$ vertices are split this way, we run the degree-$4$ algorithm. It is intuitive to think that this splitting of the vertices could increase the runtime complexity of the degree-$4$ algorithm, due to $n$ being larger. However, the addition of a forced edge incident to every new vertex means that we do not need to create more branches in the backtracking tree in order to include the new vertex in the Hamiltonian cycle. As a result, the time complexity of the degree-$4$ algorithm will remain the same.

There are $10$ unique ways of splitting a vertex of degree $5$ into one vertex of degree $3$ and another of degree $4$ connected by a forced edge. These ten ways of splitting the vertex are shown in Fig.\ \ref{fig:degree-5} for a vertex incident to edges $a,b,c,d,e$. Without loss of generality, let $a$ and $b$ be the two edges which are part of the Hamiltonian cycle. In order for $a$ and $b$ to also be part of the Hamiltonian cycle in the degree-4 graph produced, $a$ and $b$ cannot be adjacent to one another. Looking at Fig.\ \ref{fig:degree-5}, the split is successful in six of the ten ways of splitting the vertex.

If there are $f$ vertices of degree $5$, then there are $10^f$ possible ways of splitting all such vertices, of which $6^f$ will give the correct solution to the TSP. We can apply D\"urr and H\o yer's quantum algorithm for finding the minimum~\cite{durr1996} to find a splitting that leads to a shortest Hamiltonian cycle, or reporting if no cycle exists, after $O((10/6)^{f/2})$ repeated calls to the degree-4 algorithm. To ensure that the failure probability of the whole algorithm is at most $1/3$, we need to reduce the failure probability of the degree-4 algorithm to $O((10/6)^{-f/2})$, which can be achieved by repeating it $O(f)$ times and returning the minimum-length tour found. The overall runtime is thus
\begin{align*}
&O^*\left(\left(\frac{10}{6}\right)^{\frac{f}{2}}1.301^n\log L \log \log L\right)\\
 = &O^*(1.680^n\log L \log \log L).
\end{align*}
\end{proof}
It is also possible to split a vertex of degree $5$ into three vertices of degree $3$ connected by two forced edges. There are $15$ ways of performing this splitting, of which $6$ will succeed. Applying the degree-$3$ algorithm to these reduced graphs finds a runtime of
\begin{align*}
&O^*\left(\left(\frac{15}{6}\right)^{\frac{f}{2}}1.110^n\log L \log \log L\right)\\
 = &O^*(1.754^n\log L \log \log L)
\end{align*}
\noindent which performs worse than Theorem \ref{thm:deg5}. We next turn to degree-6 graphs, for which the argument is very similar.

\begin{theorem}
There is a quantum algorithm which solves the TSP for degree-$6$ graphs with failure probability $1/3$ in time $O^*(1.680^n\log L \log \log L)$.
\end{theorem}

\begin{proof}
We can extend the idea of Theorem \ref{thm:deg5} to degree-6 graphs by splitting vertices of degree $6$ into two vertices of degree $4$ connected by a forced edge. Because the degree of both new vertices is $4$, there are $\binom{6}{3}/2 = 10$ unique ways of partitioning the edges, of which 4 will fail. We show this in Fig.\ \ref{fig:degree-5} by including the dashed edge $f$ as the sixth edge. The overall runtime is the same as the degree-$5$ case.
\end{proof}

\subsection{Degree-7 graphs}

We finally considered extending the algorithm to degree-7 graphs by partitioning degree-7 vertices into one of degree $5$ and another of degree $4$, connected by a forced edge. We can split a vertex of degree $7$ into a vertex of degree $4$ and another of degree $5$ in $\binom{7}{4} = 35$ ways, of which $\binom{7-2}{4-2} + \binom{7-2}{3-2} = 15$ will not preserve the shortest Hamiltonian cycle. We then use the same process as for the degree-5 and degree-6 case, halting after $O((35/20)^{k/2})$ iterations and returning either the shortest Hamiltonian cycle found or reporting if no Hamiltonian cycle exists. From this, our overall runtime is
\begin{align*}
&O^*\left(\left(\frac{35}{20}\right)^{k/2}1.680^n\log L \log \log L\right)\\
=&O^*(2.222^n\log L \log \log L).
\end{align*}
This is the point where we no longer see a quantum speedup over the fastest classical algorithms using this approach, as classical algorithms such as those of Held-Karp~\cite{held1962} and Bj{\"o}rklund et al.~\cite{bjorklund2008} run in $O^*(2^n)$ and $O^*(1.984^n)$ time, respectively.

\section*{Note added}

Following the completion of this work, Andris Ambainis informed us of two new related results in this area. First, a quantum backtracking algorithm whose runtime depends only on the number of tree vertices visited by the classical backtracking algorithm, rather than the whole tree \cite{ambainis2016a}. This alleviates one, though not all, of the limitations of the backtracking algorithm discussed in Section II. Second, a new quantum algorithm for the general TSP based on accelerating the Held-Karp dynamic programming algorithm \cite{ambainis2016b}. The algorithm's runtime is somewhat worse than ours for graphs of degree at most 6, and it uses exponential space; but it works for any graph, rather than the special case of bounded-degree graphs considered here.

\begin{acknowledgments}
DJM was supported by the Bristol Quantum Engineering Centre for Doctoral Training, EPSRC grant EP/L015730/1. AM was supported by EPSRC Early Career Fellowship EP/L021005/1. We would like to thank Andris Ambainis for bringing refs.~\cite{ambainis2016a, ambainis2016b} to our attention.
\end{acknowledgments}

\bibliography{tsp}

\begin{thebibliography}{31}%
\makeatletter
\providecommand \@ifxundefined [1]{%
 \@ifx{#1\undefined}
}%
\providecommand \@ifnum [1]{%
 \ifnum #1\expandafter \@firstoftwo
 \else \expandafter \@secondoftwo
 \fi
}%
\providecommand \@ifx [1]{%
 \ifx #1\expandafter \@firstoftwo
 \else \expandafter \@secondoftwo
 \fi
}%
\providecommand \natexlab [1]{#1}%
\providecommand \enquote  [1]{``#1''}%
\providecommand \bibnamefont  [1]{#1}%
\providecommand \bibfnamefont [1]{#1}%
\providecommand \citenamefont [1]{#1}%
\providecommand \href@noop [0]{\@secondoftwo}%
\providecommand \href [0]{\begingroup \@sanitize@url \@href}%
\providecommand \@href[1]{\@@startlink{#1}\@@href}%
\providecommand \@@href[1]{\endgroup#1\@@endlink}%
\providecommand \@sanitize@url [0]{\catcode `\\12\catcode `\$12\catcode
  `\&12\catcode `\#12\catcode `\^12\catcode `\_12\catcode `\%12\relax}%
\providecommand \@@startlink[1]{}%
\providecommand \@@endlink[0]{}%
\providecommand \url  [0]{\begingroup\@sanitize@url \@url }%
\providecommand \@url [1]{\endgroup\@href {#1}{\urlprefix }}%
\providecommand \urlprefix  [0]{URL }%
\providecommand \Eprint [0]{\href }%
\providecommand \doibase [0]{http://dx.doi.org/}%
\providecommand \selectlanguage [0]{\@gobble}%
\providecommand \bibinfo  [0]{\@secondoftwo}%
\providecommand \bibfield  [0]{\@secondoftwo}%
\providecommand \translation [1]{[#1]}%
\providecommand \BibitemOpen [0]{}%
\providecommand \bibitemStop [0]{}%
\providecommand \bibitemNoStop [0]{.\EOS\space}%
\providecommand \EOS [0]{\spacefactor3000\relax}%
\providecommand \BibitemShut  [1]{\csname bibitem#1\endcsname}%
\let\auto@bib@innerbib\@empty
\bibitem [{\citenamefont {Gr{\"o}tschel}\ \emph {et~al.}(1991)\citenamefont
  {Gr{\"o}tschel}, \citenamefont {J{\"u}nger},\ and\ \citenamefont
  {Reinelt}}]{grotschel1991}%
  \BibitemOpen
  \bibfield  {author} {\bibinfo {author} {\bibfnamefont {M.}~\bibnamefont
  {Gr{\"o}tschel}}, \bibinfo {author} {\bibfnamefont {M.}~\bibnamefont
  {J{\"u}nger}}, \ and\ \bibinfo {author} {\bibfnamefont {G.}~\bibnamefont
  {Reinelt}},\ }\href {\doibase 10.1007/BF01415960} {\bibfield  {journal}
  {\bibinfo  {journal} {Zeitschrift f{\"u}r Operations Research}\ }\textbf
  {\bibinfo {volume} {35}},\ \bibinfo {pages} {61} (\bibinfo {year}
  {1991})}\BibitemShut {NoStop}%
\bibitem [{\citenamefont {Lawler}\ \emph {et~al.}(1985)\citenamefont {Lawler},
  \citenamefont {Lenstra}, \citenamefont {Rinnooy~Kan},\ and\ \citenamefont
  {Shmoys}}]{lawler1985}%
  \BibitemOpen
  \bibfield  {author} {\bibinfo {author} {\bibfnamefont {E.~L.}\ \bibnamefont
  {Lawler}}, \bibinfo {author} {\bibfnamefont {J.~K.}\ \bibnamefont {Lenstra}},
  \bibinfo {author} {\bibfnamefont {A.~H.~G.}\ \bibnamefont {Rinnooy~Kan}}, \
  and\ \bibinfo {author} {\bibfnamefont {D.~B.}\ \bibnamefont {Shmoys}},\
  }\href {http://eu.wiley.com/WileyCDA/WileyTitle/productCd-0471904139.html}
  {\emph {\bibinfo {title} {The Traveling Salesman Problem: A Guided Tour of
  Combinatorial Optimization}}}\ (\bibinfo  {publisher} {Wiley-Interscience
  Series in Discrete Mathematics},\ \bibinfo {year} {1985})\BibitemShut
  {NoStop}%
\bibitem [{\citenamefont {Grover}(1996)}]{grover96}%
  \BibitemOpen
  \bibfield  {author} {\bibinfo {author} {\bibfnamefont {L.~K.}\ \bibnamefont
  {Grover}},\ }in\ \href {\doibase 10.1145/237814.237866} {\emph {\bibinfo
  {booktitle} {Proc. 28\textsuperscript{th} Annual ACM Symposium on Theory of
  Computing (STOC'96)}}}\ (\bibinfo {year} {1996})\ pp.\ \bibinfo {pages}
  {212--219}\BibitemShut {NoStop}%
\bibitem [{\citenamefont {Held}\ and\ \citenamefont {Karp}(1962)}]{held1962}%
  \BibitemOpen
  \bibfield  {author} {\bibinfo {author} {\bibfnamefont {M.}~\bibnamefont
  {Held}}\ and\ \bibinfo {author} {\bibfnamefont {R.~M.}\ \bibnamefont
  {Karp}},\ }\href {http://www.jstor.org/stable/2098806} {\bibfield  {journal}
  {\bibinfo  {journal} {Journal of the Society for Industrial and Applied
  Mathematics}\ }\textbf {\bibinfo {volume} {10}},\ \bibinfo {pages} {196}
  (\bibinfo {year} {1962})}\BibitemShut {NoStop}%
\bibitem [{\citenamefont {Bj{\"o}rklund}\ \emph {et~al.}(2008)\citenamefont
  {Bj{\"o}rklund}, \citenamefont {Husfeldt}, \citenamefont {Kaski},\ and\
  \citenamefont {Koivisto}}]{bjorklund2008}%
  \BibitemOpen
  \bibfield  {author} {\bibinfo {author} {\bibfnamefont {A.}~\bibnamefont
  {Bj{\"o}rklund}}, \bibinfo {author} {\bibfnamefont {T.}~\bibnamefont
  {Husfeldt}}, \bibinfo {author} {\bibfnamefont {P.}~\bibnamefont {Kaski}}, \
  and\ \bibinfo {author} {\bibfnamefont {M.}~\bibnamefont {Koivisto}},\ }in\
  \href {\doibase 10.1007/978-3-540-70575-8_17} {\emph {\bibinfo {booktitle}
  {Proc. 35\textsuperscript{th} {I}nternational {C}onference on {A}utomata,
  {L}anguages and {P}rogramming (ICALP'08)}}}\ (\bibinfo {year} {2008})\ pp.\
  \bibinfo {pages} {198--209}\BibitemShut {NoStop}%
\bibitem [{\citenamefont {Little}\ \emph {et~al.}(1963)\citenamefont {Little},
  \citenamefont {Murty}, \citenamefont {Sweeney},\ and\ \citenamefont
  {Karel}}]{little1963}%
  \BibitemOpen
  \bibfield  {author} {\bibinfo {author} {\bibfnamefont {J.~D.~C.}\
  \bibnamefont {Little}}, \bibinfo {author} {\bibfnamefont {K.~G.}\
  \bibnamefont {Murty}}, \bibinfo {author} {\bibfnamefont {D.~W.}\ \bibnamefont
  {Sweeney}}, \ and\ \bibinfo {author} {\bibfnamefont {C.}~\bibnamefont
  {Karel}},\ }\href {\doibase 10.1287/opre.11.6.972} {\bibfield  {journal}
  {\bibinfo  {journal} {Operations Research}\ }\textbf {\bibinfo {volume}
  {11}},\ \bibinfo {pages} {972} (\bibinfo {year} {1963})}\BibitemShut
  {NoStop}%
\bibitem [{\citenamefont {Padberg}\ and\ \citenamefont
  {Rinaldi}(1991)}]{padberg1991}%
  \BibitemOpen
  \bibfield  {author} {\bibinfo {author} {\bibfnamefont {M.}~\bibnamefont
  {Padberg}}\ and\ \bibinfo {author} {\bibfnamefont {G.}~\bibnamefont
  {Rinaldi}},\ }\href {\doibase 10.1137/1033004} {\bibfield  {journal}
  {\bibinfo  {journal} {SIAM Rev.}\ }\textbf {\bibinfo {volume} {33}},\
  \bibinfo {pages} {60} (\bibinfo {year} {1991})}\BibitemShut {NoStop}%
\bibitem [{\citenamefont {{Eppstein}}(2007)}]{eppstein2007}%
  \BibitemOpen
  \bibfield  {author} {\bibinfo {author} {\bibfnamefont {D.}~\bibnamefont
  {{Eppstein}}},\ }\href {\doibase 10.7155/jgaa.00137} {\bibfield  {journal}
  {\bibinfo  {journal} {Journal of Graph Algorithms and Applications}\ }\textbf
  {\bibinfo {volume} {11}},\ \bibinfo {pages} {61} (\bibinfo {year}
  {2007})}\BibitemShut {NoStop}%
\bibitem [{\citenamefont {Iwama}\ and\ \citenamefont
  {Nakashima}(2007)}]{iwama07}%
  \BibitemOpen
  \bibfield  {author} {\bibinfo {author} {\bibfnamefont {K.}~\bibnamefont
  {Iwama}}\ and\ \bibinfo {author} {\bibfnamefont {T.}~\bibnamefont
  {Nakashima}},\ }in\ \href {\doibase 10.1007/978-3-540-73545-8_13} {\emph
  {\bibinfo {booktitle} {Proc. 13\textsuperscript{th} Annual International
  Computing and Combinatorics Conference (COCOON'07)}}}\ (\bibinfo {year}
  {2007})\ pp.\ \bibinfo {pages} {108--117}\BibitemShut {NoStop}%
\bibitem [{\citenamefont {Liśkiewicz}\ and\ \citenamefont
  {Schuster}(2014)}]{liskiewicz14}%
  \BibitemOpen
  \bibfield  {author} {\bibinfo {author} {\bibfnamefont {M.}~\bibnamefont
  {Liśkiewicz}}\ and\ \bibinfo {author} {\bibfnamefont {M.~R.}\ \bibnamefont
  {Schuster}},\ }\href {\doibase http://dx.doi.org/10.1016/j.jda.2014.02.001}
  {\bibfield  {journal} {\bibinfo  {journal} {Journal of Discrete Algorithms}\
  }\textbf {\bibinfo {volume} {27}},\ \bibinfo {pages} {1 } (\bibinfo {year}
  {2014})}\BibitemShut {NoStop}%
\bibitem [{\citenamefont {Xiao}\ and\ \citenamefont
  {Nagamochi}(2016{\natexlab{a}})}]{xiao2016degree3}%
  \BibitemOpen
  \bibfield  {author} {\bibinfo {author} {\bibfnamefont {M.}~\bibnamefont
  {Xiao}}\ and\ \bibinfo {author} {\bibfnamefont {H.}~\bibnamefont
  {Nagamochi}},\ }\href {\doibase 10.1007/s00453-015-9970-4} {\bibfield
  {journal} {\bibinfo  {journal} {Algorithmica}\ }\textbf {\bibinfo {volume}
  {74}},\ \bibinfo {pages} {713} (\bibinfo {year}
  {2016}{\natexlab{a}})}\BibitemShut {NoStop}%
\bibitem [{\citenamefont {Xiao}\ and\ \citenamefont
  {Nagamochi}(2016{\natexlab{b}})}]{xiao2016degree4}%
  \BibitemOpen
  \bibfield  {author} {\bibinfo {author} {\bibfnamefont {M.}~\bibnamefont
  {Xiao}}\ and\ \bibinfo {author} {\bibfnamefont {H.}~\bibnamefont
  {Nagamochi}},\ }\href {\doibase 10.1007/s00224-015-9612-x} {\bibfield
  {journal} {\bibinfo  {journal} {Theory of Computing Systems}\ }\textbf
  {\bibinfo {volume} {58}},\ \bibinfo {pages} {241} (\bibinfo {year}
  {2016}{\natexlab{b}})}\BibitemShut {NoStop}%
\bibitem [{\citenamefont {Bodlaender}\ \emph {et~al.}(2015)\citenamefont
  {Bodlaender}, \citenamefont {Cygan}, \citenamefont {Kratsch},\ and\
  \citenamefont {Nederlof}}]{bodlaender15}%
  \BibitemOpen
  \bibfield  {author} {\bibinfo {author} {\bibfnamefont {H.~L.}\ \bibnamefont
  {Bodlaender}}, \bibinfo {author} {\bibfnamefont {M.}~\bibnamefont {Cygan}},
  \bibinfo {author} {\bibfnamefont {S.}~\bibnamefont {Kratsch}}, \ and\
  \bibinfo {author} {\bibfnamefont {J.}~\bibnamefont {Nederlof}},\ }\href
  {\doibase http://dx.doi.org/10.1016/j.ic.2014.12.008} {\bibfield  {journal}
  {\bibinfo  {journal} {Information and Computation}\ }\textbf {\bibinfo
  {volume} {243}},\ \bibinfo {pages} {86 } (\bibinfo {year}
  {2015})}\BibitemShut {NoStop}%
\bibitem [{\citenamefont {Cygan}\ \emph {et~al.}(2011)\citenamefont {Cygan},
  \citenamefont {Nederlof}, \citenamefont {Pilipczuk}, \citenamefont
  {Pilipczuk}, \citenamefont {van Rooij},\ and\ \citenamefont
  {Wojtaszczyk}}]{cygan11}%
  \BibitemOpen
  \bibfield  {author} {\bibinfo {author} {\bibfnamefont {M.}~\bibnamefont
  {Cygan}}, \bibinfo {author} {\bibfnamefont {J.}~\bibnamefont {Nederlof}},
  \bibinfo {author} {\bibfnamefont {M.}~\bibnamefont {Pilipczuk}}, \bibinfo
  {author} {\bibfnamefont {M.}~\bibnamefont {Pilipczuk}}, \bibinfo {author}
  {\bibfnamefont {J.}~\bibnamefont {van Rooij}}, \ and\ \bibinfo {author}
  {\bibfnamefont {J.}~\bibnamefont {Wojtaszczyk}},\ }in\ \href {\doibase
  10.1109/FOCS.2011.23} {\emph {\bibinfo {booktitle} {Proc. IEEE
  52\textsuperscript{nd} Annual Symposium on Foundations of Computer Science
  (FOCS'11)}}}\ (\bibinfo {year} {2011})\ pp.\ \bibinfo {pages} {150--159},\
  \Eprint {http://arxiv.org/abs/arXiv:1103.0534} {arXiv:1103.0534} \BibitemShut
  {NoStop}%
\bibitem [{\citenamefont {Bj{\"o}rklund}(2014)}]{bjorklund14}%
  \BibitemOpen
  \bibfield  {author} {\bibinfo {author} {\bibfnamefont {A.}~\bibnamefont
  {Bj{\"o}rklund}},\ }\href {\doibase 10.1137/110839229} {\bibfield  {journal}
  {\bibinfo  {journal} {SIAM Journal on Computing}\ }\textbf {\bibinfo {volume}
  {43}},\ \bibinfo {pages} {280} (\bibinfo {year} {2014})}\BibitemShut
  {NoStop}%
\bibitem [{\citenamefont {Montanaro}(2015)}]{montanaro2015}%
  \BibitemOpen
  \bibfield  {author} {\bibinfo {author} {\bibfnamefont {A.}~\bibnamefont
  {Montanaro}},\ }\href@noop {} {\enquote {\bibinfo {title} {Quantum walk
  speedup of backtracking algorithms},}\ } (\bibinfo {year} {2015}),\ \Eprint
  {http://arxiv.org/abs/arXiv:1509.02374} {arXiv:1509.02374} \BibitemShut
  {NoStop}%
\bibitem [{\citenamefont {D{\"o}rn}(2007)}]{dorn2007}%
  \BibitemOpen
  \bibfield  {author} {\bibinfo {author} {\bibfnamefont {S.}~\bibnamefont
  {D{\"o}rn}},\ }in\ \href@noop {} {\emph {\bibinfo {booktitle} {Proc. CIE}}}\
  (\bibinfo {year} {2007})\ pp.\ \bibinfo {pages} {123--131}\BibitemShut
  {NoStop}%
\bibitem [{\citenamefont {Brassard}\ and\ \citenamefont
  {H{\o}yer}(1997)}]{brassard1997}%
  \BibitemOpen
  \bibfield  {author} {\bibinfo {author} {\bibfnamefont {G.}~\bibnamefont
  {Brassard}}\ and\ \bibinfo {author} {\bibfnamefont {P.}~\bibnamefont
  {H{\o}yer}},\ }\href@noop {} {\enquote {\bibinfo {title} {An exact quantum
  polynomial-time algorithm for {S}imon's problem},}\ } (\bibinfo {year}
  {1997}),\ \Eprint {http://arxiv.org/abs/arXiv:quant-ph/9704027}
  {arXiv:quant-ph/9704027} \BibitemShut {NoStop}%
\bibitem [{\citenamefont {D{\"u}rr}\ and\ \citenamefont
  {H{\o}yer}(1996)}]{durr1996}%
  \BibitemOpen
  \bibfield  {author} {\bibinfo {author} {\bibfnamefont {C.}~\bibnamefont
  {D{\"u}rr}}\ and\ \bibinfo {author} {\bibfnamefont {P.}~\bibnamefont
  {H{\o}yer}},\ }\href@noop {} {\enquote {\bibinfo {title} {A quantum algorithm
  for finding the minimum},}\ } (\bibinfo {year} {1996}),\ \Eprint
  {http://arxiv.org/abs/arXiv:quant-ph/9607014} {arXiv:quant-ph/9607014}
  \BibitemShut {NoStop}%
\bibitem [{\citenamefont {Mandr{\`a}}\ \emph {et~al.}(2016)\citenamefont
  {Mandr{\`a}}, \citenamefont {Guerreschi},\ and\ \citenamefont
  {Aspuru-Guzik}}]{mandra2016}%
  \BibitemOpen
  \bibfield  {author} {\bibinfo {author} {\bibfnamefont {S.}~\bibnamefont
  {Mandr{\`a}}}, \bibinfo {author} {\bibfnamefont {G.~G.}\ \bibnamefont
  {Guerreschi}}, \ and\ \bibinfo {author} {\bibfnamefont {A.}~\bibnamefont
  {Aspuru-Guzik}},\ }\href {http://stacks.iop.org/1367-2630/18/i=7/a=073003}
  {\bibfield  {journal} {\bibinfo  {journal} {New Journal of Physics}\ }\textbf
  {\bibinfo {volume} {18}},\ \bibinfo {pages} {073003} (\bibinfo {year}
  {2016})}\BibitemShut {NoStop}%
\bibitem [{\citenamefont {Marto\ifmmode~\check{n}\else \v{n}\fi{}\'ak}\ \emph
  {et~al.}(2004)\citenamefont {Marto\ifmmode~\check{n}\else \v{n}\fi{}\'ak},
  \citenamefont {Santoro},\ and\ \citenamefont {Tosatti}}]{martonak2004}%
  \BibitemOpen
  \bibfield  {author} {\bibinfo {author} {\bibfnamefont {R.}~\bibnamefont
  {Marto\ifmmode~\check{n}\else \v{n}\fi{}\'ak}}, \bibinfo {author}
  {\bibfnamefont {G.~E.}\ \bibnamefont {Santoro}}, \ and\ \bibinfo {author}
  {\bibfnamefont {E.}~\bibnamefont {Tosatti}},\ }\href {\doibase
  10.1103/PhysRevE.70.057701} {\bibfield  {journal} {\bibinfo  {journal} {Phys.
  Rev. E}\ }\textbf {\bibinfo {volume} {70}},\ \bibinfo {pages} {057701}
  (\bibinfo {year} {2004})}\BibitemShut {NoStop}%
\bibitem [{\citenamefont {Barker}(1979)}]{barker1979}%
  \BibitemOpen
  \bibfield  {author} {\bibinfo {author} {\bibfnamefont {J.~A.}\ \bibnamefont
  {Barker}},\ }\href {\doibase http://dx.doi.org/10.1063/1.437829} {\bibfield
  {journal} {\bibinfo  {journal} {The Journal of Chemical Physics}\ }\textbf
  {\bibinfo {volume} {70}},\ \bibinfo {pages} {2914} (\bibinfo {year}
  {1979})}\BibitemShut {NoStop}%
\bibitem [{\citenamefont {Metropolis}\ \emph {et~al.}(1953)\citenamefont
  {Metropolis}, \citenamefont {Rosenbluth}, \citenamefont {Rosenbluth},
  \citenamefont {Teller},\ and\ \citenamefont {Teller}}]{metropolis1953}%
  \BibitemOpen
  \bibfield  {author} {\bibinfo {author} {\bibfnamefont {N.}~\bibnamefont
  {Metropolis}}, \bibinfo {author} {\bibfnamefont {A.~W.}\ \bibnamefont
  {Rosenbluth}}, \bibinfo {author} {\bibfnamefont {M.~N.}\ \bibnamefont
  {Rosenbluth}}, \bibinfo {author} {\bibfnamefont {A.~H.}\ \bibnamefont
  {Teller}}, \ and\ \bibinfo {author} {\bibfnamefont {E.}~\bibnamefont
  {Teller}},\ }\href {\doibase http://dx.doi.org/10.1063/1.1699114} {\bibfield
  {journal} {\bibinfo  {journal} {The Journal of Chemical Physics}\ }\textbf
  {\bibinfo {volume} {21}},\ \bibinfo {pages} {1087} (\bibinfo {year}
  {1953})}\BibitemShut {NoStop}%
\bibitem [{\citenamefont {Kernighan}\ and\ \citenamefont
  {Lin}(1970)}]{kernighan1970}%
  \BibitemOpen
  \bibfield  {author} {\bibinfo {author} {\bibfnamefont {B.~W.}\ \bibnamefont
  {Kernighan}}\ and\ \bibinfo {author} {\bibfnamefont {S.}~\bibnamefont
  {Lin}},\ }\href {\doibase 10.1002/j.1538-7305.1970.tb01770.x} {\bibfield
  {journal} {\bibinfo  {journal} {The Bell System Technical Journal}\ }\textbf
  {\bibinfo {volume} {49}},\ \bibinfo {pages} {291} (\bibinfo {year}
  {1970})}\BibitemShut {NoStop}%
\bibitem [{\citenamefont {Martin}\ and\ \citenamefont
  {Otto}(1996)}]{martin1996}%
  \BibitemOpen
  \bibfield  {author} {\bibinfo {author} {\bibfnamefont {O.~C.}\ \bibnamefont
  {Martin}}\ and\ \bibinfo {author} {\bibfnamefont {S.~W.}\ \bibnamefont
  {Otto}},\ }\href {\doibase 10.1007/BF02601639} {\bibfield  {journal}
  {\bibinfo  {journal} {Annals of Operations Research}\ }\textbf {\bibinfo
  {volume} {63}},\ \bibinfo {pages} {57} (\bibinfo {year} {1996})}\BibitemShut
  {NoStop}%
\bibitem [{\citenamefont {Chen}\ \emph {et~al.}(2011)\citenamefont {Chen},
  \citenamefont {Kong}, \citenamefont {Chong}, \citenamefont {Qin},
  \citenamefont {Zhou}, \citenamefont {Peng},\ and\ \citenamefont
  {Du}}]{chen11}%
  \BibitemOpen
  \bibfield  {author} {\bibinfo {author} {\bibfnamefont {H.}~\bibnamefont
  {Chen}}, \bibinfo {author} {\bibfnamefont {X.}~\bibnamefont {Kong}}, \bibinfo
  {author} {\bibfnamefont {B.}~\bibnamefont {Chong}}, \bibinfo {author}
  {\bibfnamefont {G.}~\bibnamefont {Qin}}, \bibinfo {author} {\bibfnamefont
  {X.}~\bibnamefont {Zhou}}, \bibinfo {author} {\bibfnamefont {X.}~\bibnamefont
  {Peng}}, \ and\ \bibinfo {author} {\bibfnamefont {J.}~\bibnamefont {Du}},\
  }\href {\doibase 10.1103/PhysRevA.83.032314} {\bibfield  {journal} {\bibinfo
  {journal} {Phys. Rev. A}\ }\textbf {\bibinfo {volume} {83}},\ \bibinfo
  {pages} {032314} (\bibinfo {year} {2011})}\BibitemShut {NoStop}%
\bibitem [{\citenamefont {Belovs}(2013)}]{belovs2013}%
  \BibitemOpen
  \bibfield  {author} {\bibinfo {author} {\bibfnamefont {A.}~\bibnamefont
  {Belovs}},\ }\href@noop {} {\enquote {\bibinfo {title} {Quantum walks and
  electric networks},}\ } (\bibinfo {year} {2013}),\ \Eprint
  {http://arxiv.org/abs/arXiv:1302.3143} {arXiv:1302.3143} \BibitemShut
  {NoStop}%
\bibitem [{\citenamefont {Belovs}\ \emph {et~al.}(2013)\citenamefont {Belovs},
  \citenamefont {Childs}, \citenamefont {Jeffery}, \citenamefont {Kothari},\
  and\ \citenamefont {Magniez}}]{belovs13a}%
  \BibitemOpen
  \bibfield  {author} {\bibinfo {author} {\bibfnamefont {A.}~\bibnamefont
  {Belovs}}, \bibinfo {author} {\bibfnamefont {A.}~\bibnamefont {Childs}},
  \bibinfo {author} {\bibfnamefont {S.}~\bibnamefont {Jeffery}}, \bibinfo
  {author} {\bibfnamefont {R.}~\bibnamefont {Kothari}}, \ and\ \bibinfo
  {author} {\bibfnamefont {F.}~\bibnamefont {Magniez}},\ }in\ \href {\doibase
  10.1007/978-3-642-39206-1_10} {\emph {\bibinfo {booktitle} {Proc.
  40\textsuperscript{th} {I}nternational {C}onference on {A}utomata,
  {L}anguages and {P}rogramming (ICALP'13)}}}\ (\bibinfo {year} {2013})\ pp.\
  \bibinfo {pages} {105--122}\BibitemShut {NoStop}%
\bibitem [{\citenamefont {Dechter}(1990)}]{dechter1990}%
  \BibitemOpen
  \bibfield  {author} {\bibinfo {author} {\bibfnamefont {R.}~\bibnamefont
  {Dechter}},\ }\href {\doibase http://dx.doi.org/10.1016/0004-3702(90)90046-3}
  {\bibfield  {journal} {\bibinfo  {journal} {Artificial Intelligence}\
  }\textbf {\bibinfo {volume} {41}},\ \bibinfo {pages} {273 } (\bibinfo {year}
  {1990})}\BibitemShut {NoStop}%
\bibitem [{\citenamefont {Ambainis}\ and\ \citenamefont
  {Kokainis}(2016{\natexlab{a}})}]{ambainis2016a}%
  \BibitemOpen
  \bibfield  {author} {\bibinfo {author} {\bibfnamefont {A.}~\bibnamefont
  {Ambainis}}\ and\ \bibinfo {author} {\bibfnamefont {M.}~\bibnamefont
  {Kokainis}},\ }\href@noop {} {\enquote {\bibinfo {title} {Quantum algorithm
  for tree size estimation, with applications to backtracking and 2-player
  games},}\ } (\bibinfo {year} {2016}{\natexlab{a}}),\ \bibinfo {note}
  {forthcoming}\BibitemShut {NoStop}%
\bibitem [{\citenamefont {Ambainis}\ and\ \citenamefont
  {Kokainis}(2016{\natexlab{b}})}]{ambainis2016b}%
  \BibitemOpen
  \bibfield  {author} {\bibinfo {author} {\bibfnamefont {A.}~\bibnamefont
  {Ambainis}}\ and\ \bibinfo {author} {\bibfnamefont {M.}~\bibnamefont
  {Kokainis}},\ }\href@noop {} {}\bibinfo {howpublished} {personal
  communication} (\bibinfo {year} {2016}{\natexlab{b}})\BibitemShut {NoStop}%
\end{thebibliography}%

\end{document}